\newcommand{\CLASSINPUTbaselinestretch}{1.2}
\documentclass[12pt,draft,onecolumn]{IEEEtran}

\usepackage{geometry}
\usepackage[english]{babel}

\usepackage[latin1]{inputenc}
\usepackage{enumerate}
\usepackage{color}
\usepackage[T1]{fontenc}
\usepackage{subfigure}
\usepackage{dsfont}
\usepackage[final]{graphicx}
\usepackage[T1]{fontenc}
\usepackage{amsmath}
\usepackage{mathtools}
\usepackage{amsthm}
\usepackage{amstext}
\usepackage{amssymb}
\usepackage{mathrsfs}
\usepackage{cite}
\usepackage{mathtools}
\usepackage{tikz}
\usetikzlibrary{arrows,positioning, shapes}
\usepackage{float}

\usepackage{setspace}
\singlespace

\usepackage{hyperref}

\def\fin{{<\infty}}

\def\eps{\epsilon}
\def\hX{\hat{X}}

\def\D{{\mathcal D}}

\def\H{{\mathcal H}}

\def\X{{\mathcal X}}
\def\Y{{\mathcal Y}}
\def\M{{\mathcal M}}
\def\S{{\mathcal S}}

\def\Z{{\mathcal Z}}

\def\indep{{\perp\!\!\!\perp}}

\def\sBer{{\mathsf{Bernoulli}}}

\usepackage[font=small,labelsep=space]{caption}
\DeclareCaptionLabelSeparator{dot}{.~}
\newtheorem{example}{Example}
\newtheorem{definition}{Definition}
\newtheorem{theorem}{Theorem}

\newtheorem{lemma}{Lemma}
\theoremstyle{remark}
\newtheorem{remark}{Remark}

\usepackage{times}
\usepackage{tikz}
\usepackage{amsmath}
\usepackage{verbatim}
\usetikzlibrary{arrows,shapes}
\tikzstyle{RectObject}=[rectangle,fill=white,draw,line width=0.2mm]
\tikzstyle{line}=[draw]
\tikzstyle{arrow}=[draw, -latex]
\usetikzlibrary{decorations.pathmorphing}
\usetikzlibrary{calc,shapes, positioning}
\usepackage{graphicx}
\usepackage{caption}
\usetikzlibrary{shapes.geometric}

\IEEEoverridecommandlockouts

\allowdisplaybreaks

\begin{document}

\title{\vspace{5.5mm}Notes on Information-Theoretic Privacy}
\author{\IEEEauthorblockN{Shahab Asoodeh, Fady Alajaji, and Tam\'{a}s Linder}\\
    \IEEEauthorblockA{Department of Mathematics and Statistics, Queen's University
    \\\{asoodehshahab, fady, linder\}@mast.queensu.ca}
}
\restoregeometry
\maketitle

\begin{abstract}
%\boldmath
We investigate the tradeoff between privacy  and utility in a situation where both privacy and utility are measured in terms of  mutual information. For the binary case, we fully characterize this tradeoff in case of \emph{perfect privacy} and also give an upper-bound for the case where some privacy leakage is allowed. We then introduce a new quantity which quantifies the amount of private information contained in the observable data and then connect it to the  optimal tradeoff between privacy and utility.

\end{abstract}
\section{Introduction}
Suppose Alice has some personal information which is represented by random variable $X$ and she wants to keep this personal information as private as possible. However there exists some correlated information, represented by $Y$, observable by an advertising company and to be displayed publicly by this company. The company gets paid to send the most information about $Y$, and at the same time it does not want to violate the privacy of Alice. The question raised in this situation is then how much information about $Y$ can be displayed without breaching privacy? Hence, it is of interest to characterize such competing objectives in the form of a quantitative tradeoff. Such a characterization provides a controllable balance between utility and privacy.

Statistical studies regarding privacy were started by Warner \cite{warner} who suggested privacy-preserving methods for survey sampling. More recently, a measure known as differential privacy was introduced by Dwork et al.\ \cite{Dwork2006}. In this setting, usually the source is modelled as a \emph{database} $X=(X_1,X_2,\dots, X_n)\in\D^n$ where $\D=\{0,1\}^\ell$.  A \emph{mechanism} $\M:\D^n\to \S$, where $\S$ is a set not necessarily equal to $\D^n$, then produces the \emph{sanitized} database based on the tradeoff between accuracy and privacy. The accuracy of differential privacy is defined via a query $q:\S\to \mathcal{R}$ where $\mathcal{R}$ is some abstract set. The query can be viewed as a question about the original database $X$ that one might ask. Each query is then answered using the sanitized database $Z$. On the one hand, the data provider wants to have an accurate answer to each query, and on the other hand, the provider needs to satisfy a certain level of privacy.

In an  information-theoretic context, $\M$ is simply a Markov kernel (i.e., channel) $Q_n(\cdot|X)$ with output $Z:=\M(X)$ which takes values in $\S$. The privacy is then measured by the upper-bound of the likelihood ratio of $x$ and $x'$ with Hamming distance $1$, that is, the mechanism is called $\eps$-differentially private if $\frac{Q_n(B|x)}{Q_n(B|x')}\leq \exp(\eps)$ for all measurable $B\subset \S$ and all $x,x'\in\D^n$ such that $d_H(x,x')=1$, where $d_H$ is the Hamming distance. Note that this definition does not involve the prior distribution of $x$ and $x'$. Another measure of privacy was recently proposed under the name of \emph{a posteriori differential privacy} which incorporates the prior distributions by Wang et al.\ \cite{posteriordifferential}.

The locality requirement of $d_H(x,x')=1$ in the definition makes it hard to connect differential privacy to information theory. To overcome this problem, Duchi et al.\ \cite{privacyaware} removed the condition $d_H(x,x')=1$. This generalized definition yields the upper bound $I(X,Z)\leq \epsilon$ on the mutual information, which gives an information-theoretic interpretation of differential privacy. Hence generalized $\eps$-differentially private mechanism leaks at most $\eps$ private information.

Despite its frequent use in computer science, differential privacy does not characterize the optimal balance between privacy versus accuracy.  For example, if we want only 1\% privacy leakage, it is not clear what the best achievable accuracy is. Furthermore, it is not clear how to define differential privacy when instead of the database $X$, another database $Y$, correlated with $X$, is observable.

The problem treated in this paper can also be contrasted with the more well-studied concept of \emph{secrecy}. While in secrecy problems, e.g., in cryptography, wiretap channel problems, etc., the aim is  to keep information secret only from wiretappers, the problem treated in privacy further aims to keep the correlated source private from the intended receiver.

Although there has been no universal way of measuring privacy in the literature, in this work we follow Yamamoto \cite{yamamotoequivocationdistortion} who proposed a private source coding model. He introduced the equivocation as the conditional entropy of the private message given the observation and then defined the privacy in the system as the equivocation involved in the decoding. He then defined the rate-distortion-equivocation function as the tradeoff between utility (i.e., distortion) and privacy (i.e, equivocation).  Inspired by this work, we use the mutual information between private information $X$ and the displayed information $Z$ as the measure of privacy and also use the mutual information between the observable data $Y$ and $Z$ as utility and then define the rate-privacy function as the optimal tradeoff between these quantities. Defining utility and privacy using the mutual information gives a more intuitive measure of how much the receiver knows about $Y$ and how much of the private information is leaked to the receiver.

The paper is organized as follows. In section II, we formulate the problem in terms of the rate-privacy function and also study the binary case. We show that if zero privacy leakage is required, then in the binary case, no information from $Y$ can be transmitted.  In section III we give a multi-letter version of the rate-privacy function in a special case and show that even if $n$ different copies of $Y$ are observed, non-zero information can be transmitted about $Y$ when vanishing privacy leakage is required. In section IV, we define a new quantity related with privacy and pose an intuitive question connecting the new quantity with the rate-privacy function for the case of zero privacy leakage.
\section{Problem Formulation and the rate-privacy Function}
Consider two random variables $X\in\X$ and $Y\in \Y$ with $|\X|, |\Y|\fin$ and fixed joint distribution $P_{XY}$. $X$ is the \emph{private data} and $Y$ is the \emph{observable data} correlated with $X$. Suppose there exists a channel $P_{Z|Y}$ such that $Z$, the \emph{displayed data}, has limited information about $X$. This channel is called the \emph{privacy filter}. The objective is then to find the most informative privacy filter, i.e., a channel which preserves most of the information contained in $Y$. This setup is shown in Fig.~\ref{fig:privacy}.
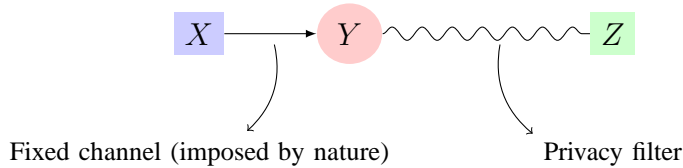
\begin{figure}[!h]
\centering
\begin{tikzpicture}
        \draw (-1,-1) node[fill=blue!20, anchor=base] (private) {$X$};
        \draw (1,-1) node[fill=red!20, ellipse,anchor=base] (public) {$Y$};
        \draw (4.5,-1) node[fill=green!20, anchor=base] (generate) {$Z$};
        \path [arrow] (private) -- (public);
        \draw[decorate, decoration = {snake, segment length = .4cm }] (public) --  (generate);
         \node[below = of private] (note1) {\small{Fixed channel (imposed by nature)}};
          \coordinate (channel1) at (0,-1);
         \draw [->] (channel1)        to [bend left] (note1);
          \node[below  = of generate ] (note2) {\small{Privacy filter}};
          \coordinate (channel2) at (3,-1);
         \draw [->] (channel2)        to [bend right] (note2);
\end{tikzpicture}
\caption{Information-theoretic privacy.} \label{fig:privacy}
\end{figure}
In particular, we are interested in characterizing the quantity,
\begin{equation}\label{gepsilon}
  g_{\epsilon}(X;Y):=\max_{P_{Z|Y}:I(X;Z)\leq \eps}I(Y;Z),
\end{equation}
which we call, the \emph{rate-privacy function}.
The dual representation of $g_{\epsilon}(X;Y)$ is given in \cite{funnel} and called the \emph{privacy funnel}. Basically, in this model, the privacy and utility are both measured using mutual information. Note that since $I(Y;Z)$ is a convex function of $P_{Z|Y}$ and furthermore the constraint set $\D_{\eps}:=\{P_{Z|Y}:~I(X;Z)\leq \eps\}$ is convex and compact, the maximum in \eqref{gepsilon} occurs at the extreme points, namely for a $P_{Z|Y}\in\D_{\eps}$ such that $I(X;Z)=\eps$.  If we restrict $P_{Z|Y}$ to be a deterministic function $f$, we get the simplified quantity
\begin{equation}\label{gtildaepsilon}
  \tilde{g}_{\epsilon}(X;Y):=\sup_{f: I(f(Y); X)= \epsilon}H(f(Y)).
\end{equation}
Using the Carath\'{e}odory-Fenchel theorem, one can readily show that it suffices that the random variable $Z$ is supported on an alphabet $\Z$ with cardinality $|\Z|\leq |\Y|+1$.

In the study of $g_{\eps}(X;Y)$ for general $P_{XY}$, the most interesting case is when $\eps=0$ (the so-called \emph{prefect privacy}), i.e., no privacy leakage is allowed.
The following theorem shows that for binary $X$ and $Y$ and an arbitrary channel between $X$ and $Y$ the requirement of  perfect privacy allows no information transfer from $Y$.
\begin{theorem}\label{theorem1}
For any pair of dependent binary random variables $X$ and $Y$, we have $$g_0(X;Y)=0.$$
\end{theorem}
\begin{proof}
  In the perfect privacy regime the constraint set reduces to $\D_0=\{P_{Z|Y}:~ Z\indep X\}$. Since $X$, $Y$ and $Z$ form the Markov chain $X\to Y\to Z$, we can write
  \begin{eqnarray} \label{setequation1}
                              P_{Z|Y}(\cdot|0)P_{Y|X}(0|1)+ P_{Z|Y}(\cdot|1)P_{Y|X}(1|1) &=& P_{Z|X}(\cdot|1) \nonumber\\
                              P_{Z|Y}(\cdot|0)P_{Y|X}(0|0)+ P_{Z|Y}(\cdot|1)P_{Y|X}(1|0) &=& P_{Z|X}(\cdot|0).\nonumber
  \end{eqnarray}
The condition $Z\indep X$ implies that $P_{Z|X}(\cdot|1)=P_{Z|X}(\cdot|0)=P_Z(\cdot)$ and hence from the above,
  \begin{eqnarray} \label{setequation2}
                              P_{Z|Y}(\cdot|0)P_{Y|X}(0|1)+ P_{Z|Y}(\cdot|1)P_{Y|X}(1|1) &=& P_Z(\cdot)\nonumber \\
                              P_{Z|Y}(\cdot|0)P_{Y|X}(0|0)+ P_{Z|Y}(\cdot|1)P_{Y|X}(1|0) &=& P_Z(\cdot).\nonumber
  \end{eqnarray}
From the assumption that $X$ and $Y$  are dependent, it follows that the above system of equations has a unique solution. The unique solution turns out to satisfy $P_{Z|Y}(\cdot|0)=P_{Z|Y}(\cdot|1)=P_{Z}(\cdot)$, which implies that $I(Y;Z)=0$.
\end{proof}
%\textcolor{blue}{
 Note that the theorem does not necessarily hold for non-binary $X$ and $Y$. In fact, it is easy to construct an example for ternary $Y$ and binary $X$ in which $g_0(X;Y)> 0$ (for instance, see Example~\ref{BECexample}). Berger and Yeung \cite[Appendix II]{berger}, gave a necessary condition for $g_0(X;Y)> 0$, in a different context.
 \begin{definition}[\cite{berger}]
 The random variable $X$ is said to be \emph{weakly independent} of $Y$ if the rows of the transition matrix $P_{X|Y}$, i.e., the set of vectors $\{P_{X|Y}(\cdot|y), ~y\in\mathcal{Y}\}$, are linearly dependent.
 \end{definition}
 In \cite{berger}, it is proved that if $X$ is weakly independent of $Y$ then there exists a binary random variable $Z$ such that $Z\indep X$ which is correlated with $Y$, and hence $g_0(X;Y)> 0$. This condition is met, for example, if $|\Y|>|\X|$. It is also straightforward to show that this condition is indeed a necessary and sufficient for $g_0(X;Y)> 0$. It is straightforward to see that if $Y$ is binary then $X$ is weakly independent of $Y$ if and only if $X$ and $Y$ are independent. This together with the fact that weak independence is a necessary and sufficient condition for $g_0(X;Y)>0$, imply the following lemma which generalizes Theorem~\ref{theorem1}.
 \begin{lemma}\label{generalizedtheorem1}
Let $Y$ be a binary random variable. Then $g_0(X;Y)$ is equal to either $H(Y)$ or zero.
\end{lemma}   %}
\section{A Multi-letter Version of $\tilde{g}_{\eps}(X;Y)$}
We next consider the simplified version of the rate-privacy function $\tilde{g}_{\epsilon}(X;Y)$ defined in \eqref{gtildaepsilon}, in the limit when $\eps\to 0$. Suppose for any $x\in\X$, inducing the distribution $P_{Y|X}(\cdot|x)$ over $\Y$, one takes $n$ independent copies of $Y$ with distribution $P_{Y^n|X}(y^n|x)=\prod_{i=1}^nP_{Y|X}(y_i|x)$. The privacy constraint is that; $I(f(Y^n); X)= \epsilon$ for every $n$ and every deterministic function $f:\Y^n\to \Z$ where $|\Z|\leq |\Y|$. Let $\tilde {g}_{n, \epsilon}(X;Y)$ denote  $\frac{1}{n}\tilde {g}_{\epsilon}(X;Y^n)$ when the distribution $P_{Y^n|X}$ is specified as above, so that
\begin{equation}\label{relaxed}
  \tilde {g}_{n, \epsilon}(X;Y):=\frac{1}{n}\sup_{f:~I(f(Y^n);X)= \epsilon}H(f(Y^n)).
\end{equation}
The following theorem gives an asymptotic lower bound on $\tilde {g}_{n, \epsilon}(X;Y)$.
\begin{theorem}\label{theorem2}
For any pair of random variables $(X,Y)$ with fixed joint distribution $P_{XY}$, we have
$$\lim_{\epsilon\to 0}\lim_{n\to \infty}\frac{1}{n}\tilde {g}_{n, \epsilon}(X;Y)\geq H^*_{\infty}(Y|X),$$
where  the min-entropy is defined as
\begin{equation}\label{hinfinity}
  H^*_{\infty}(Y|X):=\min_{x\in\mathcal{X}}\min_{y\in\mathcal{Y}}\left(-\log P_{Y|X}(y|x)\right).
\end{equation}
\end{theorem}
\begin{proof}
Suppose $|\X|=m$ and each $x_j\in\X$, $j=1,\dots,m$, induces the product distribution $P_{j}^n(y^n):=P_{Y^n|X}(y^n|x_j)=\prod_{k=1}^nP_{Y|X}(y_k|x_j)$ over $\Y$. Given these $m$ distributions $P_j^n$ for $j=1, 2, \dots, m$, we construct nearly equiprobable bins $K^n_j(i)\subset \Y^n$ for $i=1, 2, \dots, 2^r$, (with $r$ to be determined later), such that $P_j^n(K^n_j(i)):=P_j^n(Y^n\in K_j^n(i))$ is close to $2^{-r}$ for each $j=1, 2, \dots, m$ and $i=1, 2, \dots, 2^r$. Let $U^r$ denote the uniform distribution over $\{0,1\}^r$ and $V(P,Q)$ denote the total variation distance between distributions $P$ and $Q$.\\ \indent
Note that each bin $K_j^n(i)$ is an agglomeration of some mass points of $P_j^n(y^n)$ for each $j=1, 2, \dots, m$ and therefore the probability of each bin is equal to the sum of the probabilities of points $y^n$ it contains.
%The definition of min-entropy asserts that $P_j^n(y^n)\leq 2^{-H_{\infty}(Y^n|x_j)}=2^{-nH_{\infty}(Y|x_j)}$ for $j=1, 2, \dots, m$.
Recalling the definition of $H_{\infty}^*(Y|X)$ in \eqref{hinfinity}, we can write
\begin{equation}\label{assumption1}
  P_j^n(y^n)\leq 2^{-nH_{\infty}^*(Y|X)}, \qquad j=1, 2, \dots, m.
\end{equation}
We start the construction of the bins $K^n_j(1), K_j^n(2), \dots, K_j^n(J_j)$ for each $j=1, 2, \dots, m$ where $J_j\leq 2^r-1$ is the number of bins for each $j$. The first bin is constructed as follows. We agglomerate the minimal number of mass points of $P_j^n$ into $K_j^n(1)$ as needed to make sure
\begin{equation}\label{lowerbound}
  P_j^n(K_n(1))\geq 2^{-r}-2^{-s},
\end{equation}
for some $s<nH_{\infty}^*(Y|X)$. This together with \eqref{assumption1} shows that
\begin{equation}\label{upperbound}
  P_1^n(K_n(1))< 2^{-r}-2^{-s}+2^{-nH_{\infty}^*(Y|X)},
\end{equation}
which can be simplified as
\begin{equation}\label{upperbound2}
  P_1^n(K_n(1))< 2^{-r},
\end{equation}
because $s<nH_{\infty}^*(Y|X)$. \\ \indent
Once condition \eqref{lowerbound} is met, the construction for the first bin is completed and we move on to the second bin. This procedure can go on until either we run out of mass points or the restriction  $J_j\leq 2^r-1$ is violated. In the latter case, we set $J_j=2^r-1$ and then collect all mass points left into the bin $K_j^n(J_j+1)$. The former happens if the total probability of the left-over is strictly less than $ 2^{-r}-2^{-s}$ so that we can not meet the requirement \eqref{lowerbound}, in other words,
\begin{equation}\label{case1}
    P_j^n\left(\bigcup_{i=1}^{J_j} K_j^n(i)\right)>1- 2^{-r}+2^{-s}.
\end{equation}
On the other hand, we know from \eqref{upperbound2} that $P_j^n\left(\bigcup_{i=1}^{J_j} K_j^n(i)\right)<J_j2^{-r}$ which, together with \eqref{case1}, implies
\begin{equation}\label{boundforcase1}
    1- 2^{-r}+2^{-s}<P_j^n\left(\bigcup_{i=1}^{J_j} K_j^n(i)\right)<J_j2^{-r},
\end{equation}
leading to a lower bound for the number of bins in this case
\begin{equation}\label{Jlowerbound}
   J_j>2^r+2^{r-s}-1,
\end{equation}
which is greater than the allowable upper-bound $2^r-1$. We can hence conclude that with $s$ that satisfies $s<nH_{\infty}^*$, the procedure stops only when the restriction $J_j\leq 2^r-1$ is violated, and therefore, we assume $J_j=2^r-1$ in what follows.\\ \indent
As specified earlier, we construct the last bin $K(J_j+1)$ by including all the leftover mass there. We therefore have
\begin{equation}\label{lastbin}
  K_j^n(J_j+1)=\text{supp}\{P_j^n\}-\bigcup_{i=1}^{J_j}K_j^n(i),
\end{equation}
where $\text{supp}\{P_j^n\}$ denotes the support of $P_j^n$.
Since each bin has probability lower-bounded by \eqref{lowerbound}, it follows from \eqref{lastbin} that
\begin{equation}\label{boundforlastbin}
P_j^n(K(J_j+1))=1-\sum_{i=1}^{J_j} P_j^n(K_j^n(i))\leq 1-J_j\left(2^{-r}-2^{-s}\right),
\end{equation}
which, after substituting $J_j=2^r-1$, is simplified as
\begin{equation}\label{boundforlastbin2}
  P_j^n(K(J_j+1))\leq 2^{r-s}+2^{-r}-2^{-s}.
\end{equation}

So far we have constructed $m\times 2^r$ bins, namely $2^r$ bins for each $P_j^n$, $j=1,2, \dots, m$. Consider now the deterministic mapping $g_n:\mathcal{Y}^n\times \mathcal{X}\to\{0,1\}^r$ defined as follows:
$$g_n(y^n, x_j)= i \qquad ~\text{if} \qquad y^n\in K_j^n(i).$$
This mapping requires $x_j$ because for each $j\in \{1,2, \dots, m\}$ the corresponding bins are disjoint. However, we know that by using a proper channel encoding and decoding, $\phi_n$ and $\psi_n$, respectively, one can decode $Y^n$ to obtain $\psi_n(Y^n)$ such that $P(X\neq \psi_n(Y^n))$ decays exponentially. So, we can have a deterministic function which acts only on $Y^n$ from which $x_j$ is obtained with probability exponentially close to one. Hence our sequence of deterministic mappings is:
$$f_n(y^n):=g_n(y^n, \psi_n(y^n))=i\qquad ~\text{if} \qquad y^n\in K_j^n(i).$$
where $j$ is the index of the decoded symbol, that is the $j$ such that $\psi_n(y^n)=x_j$.\\ \indent
Now let us look at the total variation distance between $\tilde{P}_j^n:=f_n\circ P^n_j$ and $U^r$ which is the uniform distribution over the set $\{0,1\}^r$.
\begin{eqnarray}
  V\left(\tilde{P}_j^n, U^r\right) &=& \sum_{i=1}^{2^r}|2^{-r}-P_j^n(K_j^n(i))|\nonumber\\
   &= &  \sum_{i=1}^{J_j}\left(2^{-r}-P_j^n(K_j^n(i))\right)\nonumber\\
   &&  +|2^{-r}-P_j^n(K_j^n(J_j+1))|\label{firstline}\\
 & \leq &  \sum_{i=1}^{J_j}2^{-s}+2^{-r}+ P_j^n(K_n(J_j+1))\label{secondline}\\
& \leq & J_j2^{-s} + 2^{-r} + 2^{r-s}+2^{-r}-2^{-s}\label{thirdline}\\
&=& 2\left(2^{r-s}+2^{-r}-2^{-s}\right)<2\left(2^{r-s}+2^{-r}\right)\nonumber.
\end{eqnarray}
where in \eqref{firstline} we use \eqref{upperbound2}, in \eqref{secondline} we use the triangle inequality and \eqref{lowerbound} and the inequality in \eqref{thirdline} follows from  \eqref{boundforlastbin2}. To make sure that $V(\tilde{P}_j^n, U^r)$ goes to zero as $n\to\infty$, we set $r=nH_{\infty}^*(Y|X)-n\delta$ and $s=nH_{\infty}^*(Y|X)-n\frac{\delta}{2}$ for some $0<\delta\leq \frac{2}{3}H_{\infty}^*(Y|X)$.
%Hence we can extract $nH_{\infty}^*(Y|X)$ many $\epsilon-$uniform bits from $Y^n$.
Hence we can make sure that  $\tilde{P}_j^n$ and $\tilde{P}_k^n$ for $j\neq k$ are at most  $2\epsilon$-distant in the total variation sense. This is because, for large $n$
$$V(\tilde{P}_j^n, \tilde{P}_k^n)\leq V(\tilde{P}_j^n, U^r)+V(\tilde{P}_k^n, U^r)\leq 2\epsilon.$$
Note that, letting $E_X[\cdot]$ denote the expectation with respect to $X$,  we have in general
\begin{eqnarray*}
  V\left(P_{ZX}, P_ZP_X\right) &=&E_X\left[V\left(P_{Z|X}(\cdot|X),P_Z\right)\right], \\
   &=& E_X\big[V\left(P_{Z|X}(\cdot|X),E_X[P_{Z|X}(\cdot|X)]\right)\big],
\end{eqnarray*}
and hence by Jensen's inequality \begin{eqnarray*}
  V(P_{ZX}, P_ZP_X)\leq \sum_{x}\sum_{x'}P_X(x)P_X(x') V\left(P_{Z|X}(\cdot|x), P_{Z|X}(\cdot|x')\right)
\end{eqnarray*}
We can therefore conclude that  $V(\tilde{P}_j^n, \tilde{P}_k^n)\leq 2\epsilon$ for all $j\neq k$ results in the following $$V\left(P_{Z_nX},P_{Z_n}P_X\right)\leq 2\epsilon,$$ where $Z_n=f_n(Y^n)$. In other words, $Z_n$ and $X$ are "$2\epsilon$-independent" for sufficiently large $n$ in sense of total variation distance. Invoking \cite[Lemma 2.7]{csiszarbook}, the theorem follows.
%%%%%%%%%%%%%%%%%%%%%%%%%%%%%%%%%%%%%%%%%%%%%%%%%%%%%%%%%%%%%%%%%%%%%%%%%%%%%%%%%%%%%%%%%%%%%%
\end{proof}

This theorem implies that, unlike in the binary case studied in Theorem~\ref{theorem1}, one can have information transfer at a positive rate while allowing perfect privacy only in the limit instead of requiring absolutely zero privacy leakage.

\section{Non-Private Information vs. the Rate-Privacy Function}
Conceptually, $g_{\eps}(X;Y)$ quantifies the "largest" part of $Y$ which carries $\eps$ amount of information about $X$. Witsenhausen \cite{Witsenhausen:valuesandbounds} defined the \emph{private information} of a pair of random variables $(X,Y)$ as
\begin{equation}\label{privateinformation}
  M(X;Y):=\max_{W:~X\to W\to Y}H(X,Y|W).
\end{equation}
Wyner \cite{wynerCI} defined the \emph{common information} of $X$ and $Y$ as
\begin{equation}\label{wynerCI}
  C_W(X;Y):=\min_{W:~X\to W\to Y}I(X,Y;W).
\end{equation}

Clearly, the definition of private information in \eqref{privateinformation} implies  $C_W(X;Y)=H(X,Y)-M(X;Y)$. Operationally, $M(X;Y)$ is the rate of information that one needs to transmit over two "non-common" channels when $C_W(X;Y)$ is transmitted over the common channel in order to be able to decode $X$ and $Y$ with arbitrarily small error probability. This definition is not immediately useful in our setting, as it is symmetric in $X$ and $Y$. We seek an asymmetric definition for the private information that $Y$ contains, i.e., the rate of information contained in $Y$ which correlates with $X$. Inspired by Wyner's common information, $C_W(X;Y)$, and G\'{a}cs-K\"{o}rner's common information \cite{gacskornerCI}, denoted by $C_{GK}(X;Y)$, we define the \emph{private information about $X$ carried by $Y$} as follows
\begin{equation}\label{privateinformation2}
C_{X}(Y):=\min_{\small{\substack{W: X\to W\to Y\\H(W|Y)=0}}}H(W),
\end{equation}
and similar to the connection between $C_W(X;Y)$ and $M(X;Y)$, we define $D_X(Y):=H(Y)-C_X(Y)$ and call it the \emph{non-private information about $X$ carried by $Y$}. The quantity $C_X(Y)$ as defined above is similar to the so-called \emph{necessary conditional entropy}, $H(Y\dagger X)$, defined by Cuff et al.\ \cite{coordinationcapacity} as $\min H(W|X)$ where the minimum is taken over $W$ that satisfies the same conditions as in \eqref{privateinformation2}. Conceptually, we decompose the information contained in $Y$ into two parts, namely, one part which correlates with $X$, denoted by $C_X(Y)$, and another part which has no correlation with $X$, denoted by $D_X(Y)$. Using the assumption $H(W|Y)=0$ in \eqref{privateinformation2}, we can obtain the following variational representation for $D_X(Y)$:
\begin{equation}\label{privateinformation3}
D_{X}(Y)=\max_{\small{\substack{W: X\to W\to Y\\H(W|Y)=0}}}H(Y|W).
\end{equation}
\begin{remark}\label{remarkonequivalence}
Since $H(W|Y)=0$ implies that $W$ is a function of $Y$, one can show that the constraint in the above maximization, i.e., the conditions  $X\to W\to Y$ and $H(W|Y)=0$, is equivalent to the "double Markov relations" $X\to W\to Y$ and $X\to Y\to W$.
\end{remark}
The so called \emph{exact common information} has been introduced in \cite{exactCI} and shown to be related to the problem of \emph{exact} generation of a joint distribution $P_{XY}$. The exact common information is defined as the minimum rate $R^*$ at which an external randomness must be supplied to physically separated agents, each responsible for one of the marginals via the private randomness, so that they are able to \emph{exactly} reproduce joint distribution $P_{XY}$, in an asymptotic formulation. As illustrated in Fig.~\ref{fig:exactgeneration}, the exact common information is the minimum rate of generating $W$ such that two independent processors construct $\hat{X}$ and $\hat{Y}$, using $W$ as an input of separate stochastic decoders, such that  $P_{\hX\hat{Y}}=P_{XY}$.
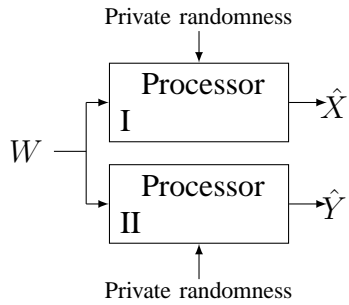
\begin{figure}[H]
\centering
\begin{tikzpicture}
         \draw (0.3, -0.65)  node[text width=2.1cm,draw](node1) {~~Processor I};
         \draw (0.3, -2)  node[text width=2.1cm,draw](node2) {~~Processor II};
         \node at (-2,-1.3) (w) {$W$};
         \coordinate (point1) at (-1.2,-1.3);
         \coordinate (point2) at (-1.2,-0.65);
         \coordinate (point3) at (-1.2,-2);

         \path [arrow] (point2) -- (node1);
         \path [arrow] (point3) -- (node2);
         \path [line] (point2) -- (point3);
         \path [line]  (w)--(point1);
         \path [arrow] (node1.east) -- (2, -0.65);
         \path [arrow] (node2.east) -- (2, -2);
         \node at (2.1,-0.65) (x) {$\hX$};
         \node at (2.1,-2) (y) {$\hat{Y}$};
         \path [arrow] (0.3, 0.3) -- (node1.north);
         \path [arrow] (0.3, -3) -- (node2.south);
         \node at (0.3,0.5) (p1) {\footnotesize{Private randomness}};
         \node at (0.3,-3.15) (p2) {\footnotesize{Private randomness}};
\end{tikzpicture}
\caption{Exact distribution generation.} \label{fig:exactgeneration}
\end{figure}

A new quantity is then introduced in \cite{exactCI}, so called \emph{common entropy} defined by
\begin{equation}\label{commonentropy}
  G(X;Y):=\min_{W: X\to W\to Y}H(W),
\end{equation}
and shown that $R^*=\lim_{n\to\infty}\frac{1}{n}G(X^n;Y^n)$. %Our quantity is also related to common entropy in that $C_X(Y)$ is the common entropy contained in $Y$.
Operationally, $C_{X}(Y)$ is the exact common information for a setting similar to Fig.~\ref{fig:exactgeneration}, except that the common input to each processor is assumed to be a deterministic function of $Y$ as depicted in Fig.~\ref{fig:exactgeneration2}.
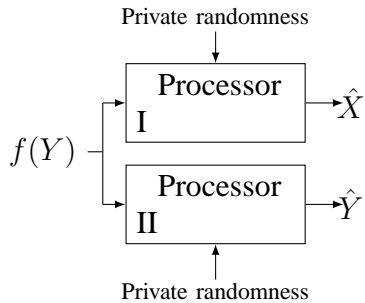
\begin{figure}[h!]
\centering
\begin{tikzpicture}
         \draw (0.3, -0.65)  node[text width=2.1cm,draw](node1) {~~Processor I};
         \draw (0.3, -2)  node[text width=2.1cm,draw](node2) {~~Processor II};
         \node at (-2,-1.3) (w) {$f(Y)$};
         \coordinate (point1) at (-1.2,-1.3);
         \coordinate (point2) at (-1.2,-0.65);
         \coordinate (point3) at (-1.2,-2);

         \path [arrow] (point2) -- (node1);
         \path [arrow] (point3) -- (node2);
         \path [line] (point2) -- (point3);
         \path [line]  (w)--(point1);
         \path [arrow] (node1.east) -- (2, -0.65);
         \path [arrow] (node2.east) -- (2, -2);
         \node at (2.1,-0.65) (x) {$\hX$};
         \node at (2.1,-2) (y) {$\hat{Y}$};
        \path [arrow] (0.3, 0.3) -- (node1.north);
         \path [arrow] (0.3, -3) -- (node2.south);
         \node at (0.3,0.5) (p1) {\footnotesize{Private randomness}};
         \node at (0.3,-3.15) (p2) {\footnotesize{Private randomness}};
\end{tikzpicture}
\caption{Exact asymmetric distribution generation.} \label{fig:exactgeneration2}
\end{figure}

\subsection{Properties of $C_{X}(Y)$}
\begin{enumerate}
\item For any $(X,Y)$ with joint distribution $P_{XY}$, we have
    \begin{equation}\label{property1}
        I(X;Y)\leq C_W(X;Y)\leq G(X;Y)\leq C_X(Y)\leq H(Y).
    \end{equation}
\begin{proof}
The first and the second inequalities are shown respectively in \cite{wynerCI} and \cite{exactCI}. The third one becomes clear once we examine the definitions of $G(X;Y)$ and $C_X(Y)$. Indeed, the objective functions in the minimization are equal, however, the constraint set for $C_X(Y)$  is a subset of the constraint set for $G(X;Y)$.   The last inequality follows from the fact that $Y$ belongs to the constraint set as well.
 \end{proof}
~~~Note that $C_X(Y)=I(X;Y)$ implies that $C_W(X;Y)=I(X;Y)=C_X(Y)$. It is a well-known fact that $C_W(X;Y)=I(X;Y)$ is equivalent to $C_{GK}(X;Y)=I(X;Y)$. Thus, $C_X(Y)=I(X;Y)$ implies that $C_{GK}(X;Y)=C_W(X;Y)$. As Wyner \cite[p. 166]{wynerCI} pointed out, these two notions of common information are equal if and only if  it is possible to write $X=(X', V)$ and $Y=(Y', V)$ such that $X'$ and $Y'$ are conditionally independent given $V$. Hence $C_X(Y)=I(X;Y)$ implies this decomposition. For the converse, suppose that we have the decomposition $X=(X', V)$ and $Y=(Y', V)$ such that $X'\to V\to Y'$. It is easy to show that for any random variable $W$ that satisfies $X\to W\to Y$ and $H(W|Y)=0$, there exits a deterministic function $f$ such that $V=f(W)$ with probability one. Hence, on the one hand, $$\max_{\substack{W: X\to W\to Y\\H(W|Y)=0}}H(Y|W)\leq H(Y|V),$$ and on the other hand, since $V$ also satisfies both conditions of $W$, we have $$\max_{\substack{W: X\to W\to Y\\H(W|Y)=0}}H(Y|W)\geq H(Y|V),$$ and therefore, $D_X(Y)=H(Y|V)=H(Y|X)$ and consequently $C_X(Y)=I(X;Y)$.
%The following lemma gives the converse.
%\begin{lemma}
%Let $X=(X',V)$ and $Y=(Y',V)$ where $X'\to V\to Y'$. Then for any $W$ such that $X\to W\to Y$ and $H(W|Y)=0$, there exists a deterministic function, $\tilde{f}$ such that $V=\tilde{f}(w)$.
%\end{lemma}
%\begin{proof}
%Assume $H(W|Y)=0$ so that $W=f(Y',V)$ for a function $f$. Then if there exist $y'_1$ and $y'_2$ and $v_1\neq v_2$ such that $f(y'_1,v_1)=f(y'_2,v_2)=w$ and $\Pr(Y'=y'_1, V=v_1)>0$ and $\Pr(Y'=y'_2, V=v_2)>0$. Then on the one hand we have
%\begin{eqnarray*}
%% \nonumber to remove numbering (before each equation)
%  \Pr(X'=x', V=v_1|f(Y',V)=w) &=& \Pr(X'=x', V=v_1|Y'=y'_1, V=v_1, f(Y',V)=w),  \\
%   &=&\Pr(X'=x', V=v_1|V=v_1, f(Y', V)=w),  \\
%   &=& \Pr(X'=x'|V=v_1),
%\end{eqnarray*}
%and on the other hand,
%\begin{eqnarray*}
%% \nonumber to remove numbering (before each equation)
%  \Pr(X'=x', V=v_1|f(Y',V)=w)  &=& \Pr(V=v_1|f(Y',V)=w)\Pr(X'=x'|V=v_1, f(Y',V)=w), \\
%   &=& \Pr(X'=x'|V=v_1)\Pr(V=v_1|f(Y',V)=w), \\
%   &<& \Pr(X'=x'|V=v_1),
%  \end{eqnarray*}
%  where the last inequality is because due to the choice of $v_1$ and $v_2$, $\Pr(V=v_1|f(Y',V)=w)<1$.
%\end{proof}
\item $C_X(Y)=0$ if and only if $X\indep Y$.
\begin{proof}
Suppose $C_X(Y)=0$. By the first inequality in \eqref{property1}, we have $I(X;Y)=0$ which implies $X\indep Y$. Conversely, if $X\indep Y$, then we have the following trivial Markov chain $X\to c\to Y$ for any constant $c$. This implies $C_X(Y)=0$.
\end{proof}
\item (\emph{Data-processing inequality}) For any $U$ such that $U\to X\to Y$, we have $C_U(Y)\leq C_X(Y)$.
\begin{proof}
Let $W^*$ attain the $C_X(Y)$. Hence we have the Markov chain $U\to X\to W^*\to Y$ and also $H(W^*|Y)=0$. It then follows by the definition that $C_U(Y)\leq H(W^*)=C_X(Y)$.
\end{proof}
\end{enumerate}
\subsection{Calculation of $D_X(Y)$}
In this section we solve the maximization in the definition of $D_X(Y)$. To do this,  we need a definition which also appears in \cite{coordinationcapacity}, \cite{newdualtoGacCI} and \cite{CIandSK}.

\begin{definition}\label{defnition3}
Given two random variables $X\in\mathcal{X}$ and $Y\in\mathcal{Y}$, let $T^{\X}:\Y\to \mathcal{P}(\mathcal{X})$ be defined by $y\to P_{X|Y}(\cdot|y)$ where $\mathcal{P}(\mathcal{X})$ is the simplex of probability distribution on $\mathcal{X}$.
\end{definition}
To solve the maximization in the definition of $D_X(Y)$, we need the following two lemmas from \cite{zeroerrorinformation}.
\begin{lemma}[\cite{zeroerrorinformation}]\label{lemma1}
The random variable $T^{\X}(Y)$ satisfies the Markov chain $X\to T^{\X}(Y)\to Y$.
\end{lemma}
%\begin{proof}
%Suppose $\phi\in \P(\X)$. Then it is clear to see that when $T^{\X}(Y)=\phi$, then $P_{X|Y}(x|y)=\phi(x)$. Hence once $T^{\X}(Y)$ is known, the conditional distribution $P_{X|Y}$ is independent of $Y$. That is, $P_{X|Y,T^{\X}(Y)}(\cdot|y,\phi)=P_{X|T^{\X}(y)}(\cdot|\phi)$ for all $y\in\mathcal{Y}$ such that $T^{\X}(y)=\phi$.
%\end{proof}
This lemma shows that the random variable $T^{\X}(Y)$ is basically a sufficient statistics of $Y$ with respect to $X$.
\begin{lemma}[\cite{zeroerrorinformation}]\label{lemma2}
Let $X, Y$ and $V$ form a Markov chain, $X\to V\to Y$ and also $H(V|Y)=0$. Then there exists a deterministic function $g$, such that $T^{\X}(Y)=g(V)$ with probability one.
\end{lemma}
This lemma together with Lemma~\ref{lemma1} implies that $T^{\X}(Y)$ is the \emph{minimal} sufficient statistics of $Y$ with respect to $X$, i.e., all other sufficient statistics of $Y$ are a function of $T^{\X}(Y)$.
%\begin{proof}
%From Remark~\ref{remarkonequivalence}, the conditions $X\to V\to Y$ and $H(V|Y)=0$ are equivalent to the double Markov relations $X\to V\to Y$ and $X\to Y\to V$. If we let $T^{\X}(V)$ and $T^{\X}(Y, V)$ be defined analogously to $T^{\X}(Y)$, then from the Markov relations we see that $T^{\X}(Y)=T^{\X}(Y, V)=T^{\X}(V)$. Hence, since $H(T^{\X}(V)|V)=0$ we can conclude that $H(T^{\X}(Y)|V)=0$ and thus there exists a deterministic function $g$ such that $T^{\X}(Y)=g(V)$.
%\endproof}

The following theorem shows that $T^{\X}(Y)$ solves the minimization in the definition of $C_X(Y)$.
\begin{theorem} \label{theorem3}
For any pair of random variables $(X,Y)$ with joint distribution $P_{XY}$, we have  $$D_X(Y)=H(Y)-H(T^{\X}(Y)).$$
\end{theorem}
\begin{proof}
Since $H(Y|W)=H(Y)-H(W)$ for all $W$ that satisfies the condition $H(W|Y)=0$, we will show that $T^{\X}(Y)$ minimizes $H(W)$ over all $W$ that satisfies Markov chain  $X\to W\to Y$ and $H(W|Y)=0$. Note that Lemma~\ref{lemma1} shows that $T^{\X}(Y)$ belongs to the constraint set of the maximization in the theorem and Lemma~\ref{lemma2} shows that $T^{\X}(Y)$ has the smallest entropy among all the random variables in the constraint set. These two lemmas therefore together  imply that $H(Y|W)$ attains its maximum value at $W=T^{\X}(Y)$.
\end{proof}
As mentioned before, $C_X(Y)\leq H(Y)$. From the previous theorem we can now give the condition under which  $C_X(Y)=H(Y)$. Assume that $\Y=\{y_1,y_2,\dots,y_m\}$.
\begin{lemma}\label{nonsingulat-Lemma}
$C_X(Y)=H(Y)$ if and only if there exists no $y_1,y_2\in\Y$ such that $P_{X|Y}(\cdot|y_1)=P_{X|Y}(\cdot|y_2)$.
\end{lemma}
\begin{proof}
 From Theorem~\ref{theorem3}, it is easy to see that if for all $y\in\Y$, $P_{X|Y}(\cdot|y)$ are different, then $C_X(Y)=H(Y)$. Conversely, suppose that $W^*$ attains $C_X(Y)$ and also suppose $H(W^*)=H(Y)$. Assume that there exist $y_1$ and $y_2$ such that $P_{X|Y}(\cdot|y_1)=P_{X|Y}(\cdot|y_2)$. Then define a new random variable $\tilde{Y}$ which takes on values on set $\{y', y_3, \dots, y_m\}$ with probabilities $(P_{Y}(y_1)+P_Y(y_2), P_Y(y_3), \dots, P_Y(y_m))$. This random variable satisfies the conditions $X\to \tilde{Y}\to Y$ and $H(\tilde{Y}|Y)=0$. However, $H(\tilde{Y})< H(Y)=H(W^*)$ which contradicts the minimality of $W^*$.
\end{proof}
\subsection{Connecting $D_X(Y)$ with $g_0(X;Y)$}
Considering the definition of $C_X(Y)$, one can loosely say that all the information contained in $Y$ which is correlated with $X$ is concentrated on $T^{\X}(Y)$, and therefore $D_X(Y)$ represents the amount of information contained in $Y$ and not correlated with $X$. This suggests that $D_X(Y)$ is equal to $g_0(X;Y)$. In what follows, we study two different cases where $D_X(Y)=g_0(X;Y)$.
%
%\begin{conjecture}
%For any pair of random variables $(X,Y)$ with fixed joint distribution $P_{XY}$, we have
%$$g_0(X;Y)=D_X(Y)=H(Y)-C_X(Y).$$
%\end{conjecture}
%\textcolor{blue}{
First we look at the case when $C_X(Y)=I(X;Y)$. We previously showed that this happens if and only if there exists the decomposition $X=(X',V)$ and $Y=(Y', V)$
such that $Y'$ is conditionally independent of $X'$ given $V$. In this case, $D_X(Y)=H(Y|X)$. It is straightforward to show that in this case we have $g_0(X;Y)\leq H(Y'|V)$.
To see this, assume otherwise, that is, suppose that there exists a random variable, say, $Z$ such that $Z\indep X$ and also $I(Y;Z)>H(Y'|V)$. Since
$$I(Y; Z)=I(Y',V;Z)=I(V;Z)+I(Y'; Z|V),$$ the assumption $I(Y; Z)>H(Y'|V)$ implies
\begin{equation*}
    I(V;Z)>H(Y'|V,Z).
\end{equation*}
This contradicts our assumption that $Z\indep X=(X', V)$. Hence we conclude that $g_0(X;Y)\leq H(Y'|V)$.
One special case of this decomposition is the case studied by Wyner \cite{wynerCI} where $X'$, $Y'$ and $V$ are mutually independent.
Consider now the following deterministic function $f$ acting on $Y=(Y', V)$, defined by $f(y) = (y', 0)$.
Then we set $Z=f(Y)$ and hence the privacy filter is $P_{Y'|Y}$. By construction we have $Z\indep (V, X')$ and hence $Z\indep X$.
Note that $I(Y;Z)=H(Z)=H(Y')$. Since we showed above that $g_0(X;Y)\leq H(Y'|V)$ and since  in this special case
$H(Y'|V)=H(Y')$, one can conclude that $g_0(X;Y)=H(Y')$. Therefore, in this case the equality $D_X(Y)=g_0(X;Y)$ holds.    %}
%In the more general case when $X'\to V\to Y'$, the above argument holds except that now $Y'$ is correlated with $V$ and hence with $X$. Hence, sending $Y'$ will violate the privacy requirement. Thus we need to transmit all information contained in $Y'$ and not in $V$. Therefore, the most information that one can transmit without breaching the privacy is $H(Y'|V)$ which implies that the optimal privacy filter is $P_{Y'|Y}$. Since $H(Y'|V)=H(Y|X)$, we have $D_X(Y)=g_0(X;Y)$. Note that in this case $g_0(X;Y)=\tilde{g}_0(X;Y)$.

The second setting that we examine is the binary case. From Theorem~\ref{theorem1} we know that $g_0(X;Y)=0$ for any binary correlated $X$ and $Y$.

Suppose $\X=\Y=\{0,1\}$, $Y\sim \sBer(p)$,  $P_{X|Y}(\cdot|0)=\sBer(\alpha)$ and $P_{X|Y}(\cdot|1)=\sBer(\beta)$. The condition $H(W|Y)=0$ implies that there exists a deterministic function $f:\Y\to \mathcal{W}$ with $|\mathcal{W}|\leq |\mathcal{Y}|$ such that $W=f(Y)$ and therefore, $P_{W|Y}(w|y)=1_{\{w=f(y)\}}$. The only possible cases for $P_{W|Y}$ are
  \begin{equation*}
  P_{W|Y}=   \begin{bmatrix}
0 & 1  \\
1 & 0    \end{bmatrix}
\end{equation*}
and
  \begin{equation*}
  P_{W|Y}=   \begin{bmatrix}
1 & 0  \\
0 & 1    \end{bmatrix}
\end{equation*}
where each column corresponds to a value of $Y\in \{0,1\}$. Thus  $W\sim \sBer(p)$ or $W\sim \sBer(1-p)$. In either case, $H(W)=H(Y)$. Hence, $C_X(Y)=H(Y)$, i.e., $D_X(Y)=0$ and thus $g_0(X;Y)=D_X(Y)=0$ which is what we wanted to show. Note that this argument does not depend on the cardinality of $\mathcal{X}$. In other words, it is impossible to send any information at non-zero rate with zero privacy leakage when $|\mathcal{Y}|=2$ which is a restatement of Lemma~\ref{generalizedtheorem1}.

%\textcolor{blue}{
Although the relation $D_X(Y)=g_0(X;Y)$ holds for the two cases described above, in the following example we have $g_0(X;Y)> D_X(Y)$.
\begin{example} \label{BECexample}
\emph{Consider $X$ distributed according to $\sBer(p)$ and the binary erasure channel $P_{Y|X}$ with erasure probability $\delta$. The output alphabet is therefore ternary $\{0, e, 1\}$ where $e$ denotes the erasure. Letting $Z=f(Y)$ where $f$ maps $Y=1$ and $Y=0$ to $1$ and $e$ to $0$, we conclude that $g_0(X;Y)\geq h(\delta)$. On the other hand, $H(Y|X)=h(\delta)$ which implies that $g_0(X;Y)=h(\delta)$. Furthermore, Lemma \ref{nonsingulat-Lemma} implies that $C_X(Y)=H(Y)$ and thus $D_X(Y)=0$. Therefore, although $D_X(Y)=0$, we can extract independent information of $X$ from $Y$ with positive rate.}
\end{example}  %}

In general, one can ask under what condition on $P_{XY}$ does the relation $D_X(Y)=g_0(X;Y)$ hold?

\section{Conclusion}
In this paper we defined a new privacy-utility tradeoff where both privacy and utility are measured in terms of mutual information. The resulting rate-privacy function characterizes the best utility when the privacy leakage is required to be less than $\eps$. For the case when $\eps=0$ (perfect privacy) we calculated the rate-privacy function for the binary case. We also introduced a new quantity which quantifies the private information contained in the observable data and examined the connection between this quantity and the rate-privacy function.

\bibliographystyle{IEEEtran}
\bibliography{bibliography}

\begin{thebibliography}{10}
\providecommand{\url}[1]{#1}
\csname url@samestyle\endcsname
\providecommand{\newblock}{\relax}
\providecommand{\bibinfo}[2]{#2}
\providecommand{\BIBentrySTDinterwordspacing}{\spaceskip=0pt\relax}
\providecommand{\BIBentryALTinterwordstretchfactor}{4}
\providecommand{\BIBentryALTinterwordspacing}{\spaceskip=\fontdimen2\font plus
\BIBentryALTinterwordstretchfactor\fontdimen3\font minus
  \fontdimen4\font\relax}
\providecommand{\BIBforeignlanguage}[2]{{%
\expandafter\ifx\csname l@#1\endcsname\relax
\typeout{** WARNING: IEEEtran.bst: No hyphenation pattern has been}%
\typeout{** loaded for the language `#1'. Using the pattern for}%
\typeout{** the default language instead.}%
\else
\language=\csname l@#1\endcsname
\fi
#2}}
\providecommand{\BIBdecl}{\relax}
\BIBdecl

\bibitem{warner}
S.~L. Warner, ``Randomized response: A survey technique for eliminating evasive
  answer bias,'' \emph{Journal of the American Statistical Association},
  vol.~60, no.~39, pp. 63--69, March 1965.

\bibitem{Dwork2006}
\BIBentryALTinterwordspacing
C.~Dwork, F.~McSherry, K.~Nissim, and A.~Smith, ``Calibrating noise to
  sensitivity in private data analysis,'' in \emph{Proceedings of the Third
  Conference on Theory of Cryptography}, ser. TCC'06.\hskip 1em plus 0.5em
  minus 0.4em\relax Berlin, Heidelberg: Springer-Verlag, 2006, pp. 265--284.
  [Online]. Available: \url{http://dx.doi.org/10.1007/11681878_14}
\BIBentrySTDinterwordspacing

\bibitem{posteriordifferential}
W.~Wang, L.~Ying, and J.~Zhang, ``On the relation between identifiability,
  differential privacy and mutual-information privacy,''
  \emph{arxiv:1402.3757}, 2014.

\bibitem{privacyaware}
J.~C. Duchi, M.~I. Jordan, and M.~J. Wainwright, ``Privacy aware learning,''
  \emph{arxiv:1210.2085}, 2013.

\bibitem{yamamotoequivocationdistortion}
H.~Yamamoto, ``A source coding problem for sources with additional outputs to
  keep secret from the receiver or wiretappers,'' \emph{IEEE Trans. Inf.
  Theory}, vol.~29, no.~6, pp. 918--923, Nov 1983.

\bibitem{funnel}
A.~Makhdoumi, S.~Salamatian, N.~Fawaz, and M.~Medard, ``From the information
  bottleneck to the privacy funnel,'' \emph{arxiv/1402.1774v4}, 2014.

\bibitem{berger}
T.~Berger and R.~Yeung, ``Multiterminal source encoding with encoder
  breakdown,'' \emph{IEEE Trans. Inf. Theory}, vol.~35, no.~2, pp. 237--244,
  Mar 1989.

\bibitem{csiszarbook}
I.~Csisz\'{a}r and J.~K\"{o}rner, \emph{Information Theory: Coding Theorems for
  Discrete Memoryless Systems}.\hskip 1em plus 0.5em minus 0.4em\relax
  Cambridge University Press, 2011.

\bibitem{Witsenhausen:valuesandbounds}
H.~S. Witsenhausen, ``Values and bounds for the common information of two
  discrete random variables,'' \emph{SIAM Journal on Applied Mathematics},
  vol.~31, no.~2, pp. 313--333, 1976.

\bibitem{wynerCI}
A.~Wyner, ``The common information of two dependent random variables,''
  \emph{IEEE Trans. Inf. Theory}, vol.~21, no.~2, pp. 163--179, Mar. 1975.

\bibitem{gacskornerCI}
\BIBentryALTinterwordspacing
P.~G\'{a}cs and J.~K\"{o}rner, ``{Common information is far less than mutual
  information},'' \emph{Probl. Inform. Control}, vol.~2, no.~2, pp. 149--162,
  1973. [Online]. Available:
  \url{http://citeseer.ifi.unizh.ch/context/562456/0}
\BIBentrySTDinterwordspacing

\bibitem{coordinationcapacity}
P.~Cuff, H.~Permuter, and T.~Cover, ``Coordination capacity,''
  \emph{Information Theory, IEEE Transactions on}, vol.~56, no.~9, pp.
  4181--4206, Sept. 2010.

\bibitem{exactCI}
G.~Kumar, C.~T. Li, and A.~{El Gamal}, ``Exact common information,''
  \emph{arxiv:1402.0062v1}, 2014.

\bibitem{newdualtoGacCI}
S.~Kamath and V.~Anantharam, ``A new dual to the {G\'{a}cs}-{K\"{o}rner} common
  information defined via the {Gray-Wyner} system,'' in \emph{Communication,
  Control, and Computing (Allerton), 2010 48th Annual Allerton Conference on},
  Sept 2010, pp. 1340--1346.

\bibitem{CIandSK}
H.~Tyagi, ``Common information and secret key capacity,'' \emph{IEEE Trans.
  Inf. Theory}, vol.~59, no.~9, pp. 5627--5640, Sept. 2013.

\bibitem{zeroerrorinformation}
S.~Wolf and J.~Wultschleger, ``Zero-error information and applications in
  cryptography,'' in \emph{Information Theory Workshop, 2004. IEEE}, Oct. 2004,
  pp. 1--6.

\end{thebibliography}

\end{document}